\documentclass[11pt]{article}
\usepackage[a4paper]{geometry}
\usepackage{microtype,xstring}
\usepackage{amssymb,amsmath,amsthm}
\usepackage{booktabs,tabularx}
\usepackage{url,xspace,soul,wrapfig}
\usepackage{graphicx,hyperref}
\usepackage[font=small, hypcap=true]{subfig,caption}
\usepackage{paralist,enumerate}

\newtheorem{theorem}{Theorem}

\graphicspath{{figures/}}

\title{Hierarchical Partial Planarity}

\author{Patrizio~Angelini, Michael~A.~Bekos
\\
\medskip
\\
Wilhelm-Schickhard-Institut f\"ur Informatik, Universit\"at T\"ubingen, Germany\\
\texttt{\{angelini,bekos\}@informatik.uni-tuebingen.de}
}

\date{}

\newcommand{\ourproblem}{{\sc Hierarchical Partial Planarity}\xspace}
\newcommand{\fundamental}{primary\xspace}
\newcommand{\important}{secondary\xspace}
\newcommand{\nonimportant}{tertiary\xspace}
\newcommand{\embproblem}{{\sc Facial-Constrained Core Planarity}\xspace}

\newcommand{\Ef}{E_p}
\newcommand{\Ei}{E_s}
\newcommand{\En}{E_t}

\newcommand{\ourinstance}[1]{$G^{#1}=(V,\Ef \cup \Ei \cup \En)$\xspace}
\newcommand{\embinstance}[1]{$\langle G^{#1}=(V,E_1 \cup E_2), W \rangle$\xspace}

\newcommand{\pert}[1]{G^\textit{\scriptsize pert}_{#1}}
\newcommand{\skel}[1]{G^\textit{\scriptsize skel}_{#1}}

\newcommand{\redge}[1]{\textit{ref}(#1)}
\newcommand{\embpert}[1]{\mathcal{G}^\textit{\scriptsize pert}_{#1}}
\newcommand{\restrpert}[1]{\mathcal{H}^\textit{\scriptsize pert}_{#1}}

\newcommand{\algocase}[1]{\smallskip\noindent\textbf{\boldmath#1}}
\newcommand{\operation}{\textsc{merge-bags}\xspace}

\newcommand{\aux}{G_\textit{\scriptsize aux}}
\newcommand{\auxb}{\overline{G}_\textit{\scriptsize aux}}

\newcommand{\subskel}{H^\textit{\scriptsize skel}_\mu}
\newcommand{\subskelemb}{\mathcal{H}^\textit{\scriptsize skel}_\mu}
\newcommand{\e}[1]{e_\mu({#1})}

\begin{document}
\maketitle

\begin{abstract}
In this paper we consider graphs whose edges are associated with a degree of \emph{importance}, which may depend on the type of connections they represent or on how recently they appeared in the scene, in a streaming setting. The goal is to construct layouts of these graphs in which the readability of an edge is proportional to its importance, that is, more important edges have fewer crossings. We formalize this problem and study the case in which there exist three different degrees of importance. We give a polynomial-time testing algorithm when the graph induced by the two most important sets of edges is biconnected. We also discuss interesting relationships with other constrained-planarity problems.
\end{abstract}

\section{Introduction}
\label{sec:introduction}
Describing a graph in terms of a stream of nodes and edges, arriving and leaving at different time instants, is becoming a necessity for application domains where massive amounts of data, too large to be stored, are produced at a very high rate. The problem of visualizing graphs under this streaming model has been introduced only recently.

In particular, the first step in this direction was performed in~\cite{bbddgppsz-dtsm-12}, where the problem of drawing trees whose edges arrive one-by-one and disappear after a certain amount of steps has been studied, from the point of view of the area requirements of straight-line planar drawings. Later on, it was proved~\cite{gp-sgdfmp-13} that polynomial area could be achieved for trees, tree-maps, and outerplanar graphs if a small number of vertex movements are allowed after each update. The problem has also been studied~\cite{dr-psg-15} for general planar graphs, relaxing the requirement that edges have to be straight-line.

In this paper we introduce a problem motivated by this model, and in particular by the fact that the \emph{importance} of vertices and edges in the scene decades with time. In fact, as soon as an edge appears, it is important to let the user clearly visualize it, possibly at the cost of moving ``older'' edges in the more cluttered part of the layout, which may be unavoidable if the graph is large or dense. The idea is that the user may not need to \emph{see} the connection between two vertices, as she \emph{remembers} it from the previous steps.

Visually, one could associate the decreasing importance of an edge with its fading; theoretically, one could associate it with the fact that it becomes more acceptable to let it participate in some crossings. As a general framework for this kind of problems, we associate a weight $w(e)$ to every edge $e \in E$ and define a function $f: E \times E \rightarrow \{\texttt{YES},\texttt{NO}\}$ that, given a pair of edges $e$ and $e'$, determines whether it is allowed to have a crossing between $e$ and $e'$ based on their weights. Of course, if no assumption is made on function $f(\cdot)$, this model allows to encode instances of the NP-complete problem {\sc Weak Realizability}~\cite{DBLP:journals/jct/Kratochvil91}, in which the pairs of edges that are allowed to cross are explicitly given as part of the input. On the other hand, already the ``natural'' assumption that, if an edge $e$ is allowed to cross an edge $e'$, then it is also allowed to cross any edge $e''$ such that $w(e'') \leq w(e')$, could potentially make the problem tractable.

As a first step towards a formalization of this general idea, we introduce problem \ourproblem, which takes as input a graph $G=(V,E=\Ef \cup \Ei \cup \En)$ whose edges are partitioned into the \emph{\fundamental} edges in $\Ef$, the \emph{\important} edges in $\Ei$, and the \emph{\nonimportant} edges in $\En$. The goal is to construct a drawing of $G$ in which the \fundamental edges are crossing-free, the \important edges can only cross \nonimportant edges, while these latter edges can also cross one another. We say that any crossing that involves a \fundamental edge or two \important ones is \emph{forbidden}. We remark that this problem can be easily modeled under the general framework we described above. Namely, we can say that all edges in $\Ef$, $\Ei$, and $\En$ have weights $4$, $2$, and $1$, respectively, and function $f(\cdot)$ is such that $f(e,e') = \texttt{YES}$ if and only if $w(e)+w(e') \leq 3$.

We observe that our problem is a generalization of the recently introduced {\sc Partial Planarity} problem~\cite{DBLP:journals/comgeo/AngeliniBLDGMPT15,DBLP:conf/gd/Schaefer14}, in which the edges of a certain subgraph of a given graph must not be involved in any crossings. An instance of this problem is in fact an instance of our problem only composed of edges in $\Ef$ and $\En$.

Our main contribution is an $O(|V|^3 \cdot |\En|)$-time algorithm for \ourproblem when the graph induced by the \fundamental and the \important edges is biconnected (see Section~\ref{sec:algorithm}). Our result builds upon a formulation of the problem in terms of a \emph{constrained-planarity} problem, which we believe to be interesting in its own. Our algorithm for this constrained-planarity problem is based on the use of SPQR-trees~\cite{DBLP:journals/algorithmica/BattistaT96,DBLP:journals/siamcomp/BattistaT96}. This formulation also allows us to uncover interesting relationships with other important graph planarity problems, like {\sc Partially Embedded Planarity}~\cite{DBLP:journals/talg/AngeliniBFJKPR15,DBLP:journals/comgeo/JelinekKR13} and {\sc Simultaneous Embedding with Fixed Edges}~\cite{DBLP:reference/crc/BlasiusKR13,DBLP:journals/comgeo/BrassCDEEIKLM07} (see Section~\ref{sec:formulation}).

In Section~\ref{sec:preliminaries} we give definitions, and in Section~\ref{sec:conclusions} we conclude with open problems.

\section{Preliminaries}
\label{sec:preliminaries}
A graph $G=(V,E)$  containing neither loops nor multiple edges is \emph{simple}. We consider simple graphs, if not otherwise specified. A \emph{drawing} $\Gamma$ of $G$ maps each vertex of $G$ to a point in the plane and each edge of $G$ to a Jordan curve between its two end-points.

A drawing is \emph{planar} if no two edges cross except, possibly, at common endpoints. A planar drawing partitions the plane into connected regions, called \emph{faces}. The unbounded one is called \emph{outer face}. A graph is \emph{planar} if it admits a planar drawing.  A \emph{planar embedding} of a planar graph is an equivalence class of planar drawings that define the same set of faces and outer face. Let $H$ be a subgraph of a planar graph $G$, and let $\mathcal{G}$ be a planar embedding of $G$. We call \emph{restriction} of $\mathcal{G}$ to $H$ the planar embedding of $H$ that is obtained by removing the edges of $G \setminus H$ from $\mathcal{G}$ (and potential isolated vertices).

A graph is \emph{connected} if for any pair of vertices there is a path connecting them. A graph is \emph{$k$-connected} if the removal of $k-1$ vertices leaves it connected. A $2$- or $3$-connected graph is also referred to as \emph{biconnected} or \emph{triconnected}, respectively. 

The SPQR-tree $\mathcal{T}$ of a biconnected graph $G$ is a labeled tree representing the decomposition of $G$ into its triconnected components~\cite{DBLP:journals/algorithmica/BattistaT96,DBLP:journals/siamcomp/BattistaT96}. Every triconnected component of $G$ is associated with a node $\mu$ in $\mathcal{T}$. The triconnected component itself is referred to as the \emph{skeleton} of $\mu$, denoted by $\skel{\mu}$, whose edges are called \emph{virtual edges}. A node $\mu \in \mathcal{T}$ can be of one of four different types:%
\begin{inparaenum}[(i)]
\item \emph{S-node}, if $\skel{\mu}$ is a simple cycle of length at least~$3$;
\item \emph{P-node}, if $\skel{\mu}$ is a bundle of at least three parallel edges;
\item \emph{Q-node}, if $\skel{\mu}$ consists of two parallel edges;
\item \emph{R-node}, if $\skel{\mu}$ is a simple triconnected graph.
\end{inparaenum}
The set of leaves of $\mathcal{T}$ coincides with the set of Q-nodes, except for one arbitrary Q-node $\rho$, which is selected as the root of $\mathcal{T}$. Also, neither two $S$-nodes, nor two $P$-nodes are adjacent in~$\mathcal{T}$. Each virtual edge in $\skel{\mu}$ corresponds to a node $\nu$ that is adjacent to $\mu$ in $\mathcal{T}$, more precisely, to another virtual edge in $\skel{\nu}$. In particular, the skeleton of each node $\mu$ (except the one of $\rho$) contains a virtual edge, called \emph{reference edge} and denoted by $\redge{\mu}$, that has a counterpart in the skeleton of its parent. The endvertices of $\redge{\mu}$ are the \emph{poles} of $\mu$. The subtree $\mathcal{T}_\mu$ of $\mathcal{T}$ rooted at $\mu$ induces a subgraph $\pert{\mu}$ of $G$, called \emph{pertinent}, which is described by $\mathcal{T}_\mu$ in the decomposition. The SPQR-tree of $G$ is unique, up to the choice of the root, and can be computed in linear time~\cite{DBLP:conf/gd/GutwengerM00}.

\section{Problem Formulation \& Relationships to Other \mbox{Problems}}
\label{sec:formulation}
In this section we define a problem, called \embproblem, that will serve as a tool to solve \ourproblem and to uncover interesting relationships with other important graph planarity problems. This problem takes as input a graph $G=(V,E_1 \cup E_2)$ and a set $W \subseteq V \times V$ of pairs of vertices. Let $H$ be the subgraph of $G$ induced by the edges in $E_1$, which we call \emph{core} of $G$. The goal is to construct a planar embedding $\mathcal{G}$ of $G$ whose restriction $\mathcal{H}$ to $H$ is such that, for each pair $\langle u, v \rangle \in W$, there exists a face of $\mathcal{H}$ that contains both $u$ and $v$.

\begin{theorem}\label{th:equivalence}
Problems \embproblem and \ourproblem are linear-time equivalent.
\end{theorem}
\begin{proof}
We show how to construct in linear time an instance \embinstance{\prime} of \embproblem starting from an instance \ourinstance{} of \ourproblem. Graph $G^\prime$ has the same vertex-set $V$ as $G$. Also, we set $E_1 = \Ef$ and $E_2 = \Ei$. Finally, for each edge $(u,v) \in \En$, we add pair $\langle u,v \rangle$ to $W$. The reduction in the opposite direction is symmetric.

Suppose that $\langle G^\prime, W \rangle$ is a positive instance, and let $\mathcal{G}^\prime$ be a corresponding planar embedding of $G^\prime$. We show how to construct a drawing $\Gamma$ of $G$ not containing any forbidden crossing. First, initialize $\Gamma$ to a planar drawing of $G^\prime$ whose embedding is $\mathcal{G}^\prime$. Note that restricting $\Gamma$ to the core $H^\prime$ of $G^\prime$ yields a planar drawing $\Gamma^\prime$ of $H^\prime$ in which, for each pair $\langle u, v \rangle \in W$, there exists a face $f_{u,v}$ of $\mathcal{H}^\prime$ that contains both $u$ and $v$. This implies that it is possible to draw edge $(u,v) \in \En$ in $\Gamma^\prime$ as a curve from $u$ to $v$ lying completely in the interior of $f_{u,v}$, and hence not crossing any \fundamental edge. Repeating this operation for every pair in $W$ yields a drawing $\Gamma$ with no forbidden crossings.

Suppose that $G$ is a positive instance, and let $\Gamma$ be the corresponding drawing of $G$. We show how to construct an embedding $\mathcal{G}^\prime$ of $G^\prime$ such that for every pair $\langle u,v \rangle \in W$, vertices $u$ and $v$ lie in the same face of $\mathcal{H}^\prime$. First, note that the drawing $\Gamma_{p,s}$ induced by the edges in $\Ef$ and $\Ei$ is planar, due to the definition of \ourproblem. Also, note that $\Gamma_{p,s}$ is a planar drawing of $G^\prime$, since $E_1 = \Ef$ and $E_2 = \Ei$. Let $\mathcal{G}^\prime$ be the planar embedding of $G^\prime$ corresponding to $\Gamma_{p,s}$. Let $\mathcal{H}^\prime$ be the restriction of $\mathcal{G}^\prime$ to $H^\prime$. Consider a pair $\langle u,v \rangle \in W$ and let $e=(u,v)$ be the corresponding \nonimportant edge of $G$. Since $e$ can be drawn in $\Gamma_{p,s}$ without crossing any \fundamental edge, vertices $u$ and $v$ are incident to the same face of $\mathcal{H}^\prime$. This concludes the proof.
\end{proof}

In the following, we describe relationships between \ourproblem and other important graph planarity problems, as {\sc Partial Planarity}~\cite{DBLP:journals/comgeo/AngeliniBLDGMPT15,DBLP:conf/gd/Schaefer14} and {\sc Partially Embedded Planarity}~\cite{DBLP:journals/talg/AngeliniBFJKPR15,DBLP:journals/comgeo/JelinekKR13}, and {\sc Simultaneous Embedding with Fixed Edges}~\cite{DBLP:reference/crc/BlasiusKR13}.

In {\sc Partial Planarity}~\cite{DBLP:journals/comgeo/AngeliniBLDGMPT15}, given a non-planar graph $G=(V,E)$ and a subset $F \subseteq E$ of its edges, the goal is to compute a drawing $\Gamma$ of $G$, if any, in which the edges of $F$ are not crossed by any edge of $G$. Positive and negative results are given in~\cite{DBLP:journals/comgeo/AngeliniBLDGMPT15} if the graph induced by $F$ is a connected spanning subgraph of $G$. In~\cite{DBLP:conf/gd/Schaefer14}, the corresponding decision problem is shown to be polynomial-time solvable. By setting $\Ef=F$, $\Ei= \emptyset$, and $\En=E \setminus F$, we can model any instance of {\sc Partial Planarity} as an instance of \ourproblem. We thus have the following.

\begin{theorem}\label{th:partial-planarity}
{\sc Partial Planarity} can be reduced in linear time to \ourproblem.
\end{theorem}

In {\sc Partially-Embedded Planarity}~\cite{DBLP:journals/talg/AngeliniBFJKPR15}, given a planar graph $G$ and a planar embedding $\mathcal{H}$ of a subgraph $H$ of $G$, the goal is to determine whether $\mathcal{H}$ can be extended to a planar embedding of $G$, and to compute this embedding, if it exists. The problem is linear-time solvable~\cite{DBLP:journals/talg/AngeliniBFJKPR15} and characterizable in terms of forbidden subgraphs~\cite{DBLP:journals/comgeo/JelinekKR13}. We prove that \ourproblem can be used to encode instances of {\sc Partially-Embedded Planarity} in which $H$ is biconnected. Note that this special case is a central ingredient in the algorithm  in~\cite{DBLP:journals/talg/AngeliniBFJKPR15} for the general case. 

\begin{theorem}\label{th:partially-embedded-planarity}
{\sc Partially-Embedded Planarity} with biconnected $H$ can be reduced in quadratic time to \ourproblem.
\end{theorem}
\begin{proof}
Let $\langle G^\prime=(V,E), H, \mathcal{H} \rangle$ be an instance of {\sc Partially-Embedded Planarity} in which $H$ is biconnected. We construct an instance \embinstance{} of \embproblem on the same vertex set $V$ as $G^\prime$, as follows. Set $E_1$ contains all the edges of $E$ that are contained in $H$; set $E_2$ contains the other ones, that is, $E_2 = E \setminus E_1 $. Finally, for every pair of non-adjacent vertices $\langle u, v \rangle$ that are on the same face of $\mathcal{H}$, we add a pair $\langle u, v \rangle$ to~$W$. This last step requires quadratic time and guarantees that in the solution of \embproblem, for each face $f$ of $\mathcal{H}$, all the vertices of $f$ are incident to the same face $f'$ of the planar embedding of the core of $G$. These vertices appear in the same order along $f$ and $f'$, since $H$ is biconnected and thus this order is unique. Hence, $\langle G^\prime, H, \mathcal{H} \rangle$ is a positive instance if and only if $\langle G, W \rangle$ is.
The statement follows by Theorem~\ref{th:equivalence}.
\end{proof}

A \emph{simultaneous embedding} of two planar graphs $G_1=(V,E_1)$ and $G_2=(V,E_2)$ embeds each graph in a planar way using the same vertex positions for both embeddings; edges are allowed to cross only if they belong to different graphs (see~\cite{DBLP:reference/crc/BlasiusKR13} for a survey). Our problem is related to a well-studied version of this problem, called {\sc Simultaneous Embedding with Fixed Edges} ({\sc Sefe})~\cite{adn-osnsp-14,DBLP:journals/jda/AngeliniBFPR12,DBLP:journals/comgeo/BlasiusR15,DBLP:journals/talg/BlasiusR16,DBLP:journals/comgeo/BrassCDEEIKLM07}, in which edges that are \emph{common} to both graphs must be embedded in the same way (and hence, cannot be crossed by other edges). So in our setting, these edges correspond to the \fundamental ones. However, to obtain a solution for {\sc Sefe}, it does not suffice to assume that the \emph{exclusive} edges of $G_1$ and $G_2$ are the \important and \nonimportant ones, respectively, as we could not guarantee that the edges of $G_2$ do not cross each other. So, in some sense, our problem seems to be more related to \emph{nearly-planar simultaneous embeddings}, where the input graphs are allowed to cross, as long as they avoid some local crossing configurations, e.g., by avoiding triples of mutually crossing edges~\cite{DBLP:journals/cj/GiacomoDLMW15}. Note that the {\sc Sefe} problem has also been studied in several settings~\cite{DBLP:conf/gd/AngeliniCCLBEKK16,DBLP:journals/jgaa/BekosDKW16,DBLP:journals/jgaa/ChanFGLMS15,DBLP:journals/jgaa/ErtenK05}.
An interpretation of {\sc Partial Planarity}, which also extends to \ourproblem, in terms of a special version of {\sc Sefe}, called {\sc Sunflower Sefe}~\cite{DBLP:reference/crc/BlasiusKR13}, was already observed  in~\cite{DBLP:journals/comgeo/AngeliniBLDGMPT15}.

The algorithm we present in Section~\ref{sec:algorithm} is inspired by an algorithm to decide in linear time whether a pair of graphs admits a {\sc Sefe} if the common graph is biconnected~\cite{DBLP:journals/jda/AngeliniBFPR12}. The main part of that algorithm is to find an embedding of the common graph in which every pair of vertices that are joined by an exclusive edge are incident to the same face; so, these edges play the role of the pairs in $W$. In a second step, it checks for crossings between exclusive edges of the same graph. Since the common graph is biconnected, the existence of these crossings does not depend on the choice of the embedding.

Thus, for instances of our problem in which the core $H$ of $G$ is biconnected, we can employ the main part of the algorithm in~\cite{DBLP:journals/jda/AngeliniBFPR12} to find a planar embedding of $H$ in which every two vertices that either are joined by an edge of $E_2$ or form a pair of $W$ are incident to the same face of $H$; note that in this case it is not even needed to perform the second check for the pairs in $W$. In this paper we extend this result to the case in which $H$ is not biconnected, but it becomes so when adding the edges of $E_2$. The main difficulty here is to ``control'' the faces of $H$ by operating on the embeddings of the biconnected graph $G$ composed of $H$ and of the edges of $E_2$. In Section~\ref{sec:algorithm} we discuss the problems arising from this fact and our proposed solution.

\section{Biconnected Facial-Constrained Core Planarity}
\label{sec:algorithm}
In this section, we provide a polynomial-time algorithm for instances \embinstance{} of \embproblem in which $G$ is biconnected. Recall that the goal is to find a planar embedding $\mathcal{G}$ of $G$ so that for each pair $\langle x,y \rangle \in W$ vertices $x$ and $y$ lie in the same face of the restriction $\mathcal{H}$ of $\mathcal{G}$ to the core $H$ of~$G$.

\subsection{High-Level Description of the Algorithm}\label{se:high-level}

We first give a high-level description of our algorithm. We perform a bottom-up traversal of the SPQR-tree $\mathcal{T}$ of $G$. At each step of the traversal, we consider a node $\mu \in \mathcal{T}$ and we search for an embedding $\embpert{\mu}$ of $\pert{\mu}$ satisfying the following requirements.%
\begin{enumerate}[R.1]
\item \label{r:1} For every pair $\langle x,y \rangle \in W$ such that $x$ and $y$ belong to $\pert{\mu}$, vertices $x$ and $y$ lie in the same face of the restriction $\restrpert{\mu}$ of $\embpert{\mu}$ to the part of the core $H$ in $\pert{\mu}$.
\item \label{r:2} For every pair $\langle x,y \rangle \in W$ such that exactly one vertex, say $x$, belongs to $\pert{\mu}$, vertex $x$ lies in the outer face of $\restrpert{\mu}$ (note that $y$ belongs to $G \setminus \embpert{\mu}$).
\end{enumerate}

In general, there may exist several ``candidate'' embeddings of $\pert{\mu}$ satisfying R.\ref{r:1} and R.\ref{r:2}. If there exists none, the instance is negative. Otherwise, we would like to select one of them and proceed with the traversal. However, while it would be sufficient to select \emph{any} embedding of $\pert{\mu}$ satisfying R.\ref{r:1}, it is possible that some of the embeddings satisfying R.\ref{r:2} are ``good'', in the sense that they can be eventually extended to an embedding of $G$ satisfying both R.\ref{r:1} and R.\ref{r:2}, while some others are not. Unfortunately, we cannot determine which ones are good at this stage of the algorithm, as this may depend on the structure of a subgraph that is considered later in the traversal. Thus, we have to maintain succinct information to describe the properties of the embeddings of $\pert{\mu}$ that satisfy R.\ref{r:1} and R.\ref{r:2}, so to group these embeddings into equivalence classes. 

\begin{figure}[t]
  \centering
  \subfloat[\label{fig:non-traversable}]{\includegraphics[width=0.22\textwidth,page=1]{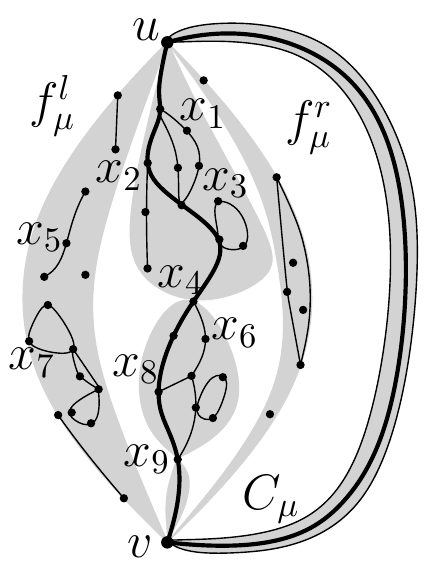}}
  \hfill
  \subfloat[\label{fig:traversable}]{\includegraphics[width=0.22\textwidth,page=2]{images/traversable.pdf}}
  \hfill
  \subfloat[\label{fig:bags}]{\includegraphics[width=0.18\textwidth,page=3]{images/traversable.pdf}}
  \hfill
  \subfloat[\label{fig:bags-traversable}]{\includegraphics[width=0.18\textwidth,page=4]{images/traversable.pdf}}
  \caption{
  (a--b) Graph $\pert{\mu}$ when $\mu$ is (a) non-traversable and (b) traversable. 
  (c--d) The bags of the nodes in (a) and in (b), respectively. A segment between $u$ and $v$ separates faces $f_\mu^l$ and $f_\mu^r$; each bag $B_\mu^i$ is represented by a circle across the segment, with its pockets $S_\mu^i$ and $T_\mu^i$ on the two sides; the vertices in the special bag $\mathcal{B}_\mu$ lie along the segment, as they are incident to both faces.}
  \label{fig:traversability}
\end{figure}

We denote by $x_1,\dots,x_k$ the vertices belonging to pairs $\langle x_i, y_i \rangle \in W$ such that $x_i \in \pert{\mu}$ and $y_i \notin \pert{\mu}$. In other words, vertices $x_1,\dots,x_k$ are those that must lie on the outer face of $\restrpert{\mu}$ due to R.\ref{r:2}. To describe the information to maintain, we need the following definition. We say that $\mu$ is \emph{non-traversable} if there is a cycle $C_\mu$ composed of edges of $E_1$ that contains both poles $u$ and $v$ of $\mu$, at least one edge of $\pert{\mu}$, and at least one of $G \setminus \pert{\mu}$; see Fig.~\ref{fig:non-traversable}. Otherwise, $\mu$ is \emph{traversable}, i.e., either in $\pert{\mu}$ or in $G \setminus \pert{\mu}$ every path between $u$ and $v$ contains edges of $E_2$; see Fig.~\ref{fig:traversable}. 

Intuitively, when $\mu$ is non-traversable, cycle $C_\mu$ splits the outer face of $\restrpert{\mu}$ into two faces $f_\mu^l$ and $f_\mu^r$ of $\mathcal{H}$ in any planar embedding of $G$. Hence, R.\ref{r:2} must be refined to take into account the possible partitions of $x_1,\dots,x_k$ with respect to their incidence to $f_\mu^l$ and $f_\mu^r$. For a single vertex $x_i \in \{x_1,\dots,x_k\}$ this is not an issue, as a flip of $\pert{\mu}$ can transform an embedding of $G$ in which $x_i$ is incident to one of these faces into another one in which it is incident to the other face. However, there may exist dependencies among different vertices of $\{x_1,\dots,x_k\}$, given by the structure of $\pert{\mu}$, which enforce the relative positions of these vertices with respect to $f_\mu^l$ and $f_\mu^r$. More precisely, let $\langle x,y \rangle, \langle x',y' \rangle \in W$ be two pairs such that $x, x' \in \pert{\mu}$ and $y, y' \notin \pert{\mu}$. Then, vertices $x$ and $x'$ may be enforced to be incident to the same face, either $f_\mu^l$ or $f_\mu^r$ (see $x_1$ and $x_3$ in Fig.~\ref{fig:non-traversable}), they may be enforced to be incident to different faces (see $x_2$ and $x_3$ in Fig.~\ref{fig:non-traversable}), or they may be independent in this respect (see $x_1$ and $x_6$ in Fig.~\ref{fig:non-traversable}).

We encode this information by associating a set of \emph{bags} with $\mu$, which contain vertices $x_1,\dots,x_k$. Each bag is composed of two \emph{pockets}; all the vertices in a pocket must be incident to the same face of $\mathcal{H}$ in any candidate embedding of $\pert{\mu}$, while all the vertices in the other pocket must be incident to the other face. Vertices of different bags are independent of each other. For the vertices of $\{x_1,\dots,x_k\}$ that are incident to both $f_\mu^l$ and $f_\mu^r$ in any embedding (see $x_4$ in Fig.~\ref{fig:non-traversable}), we add a \emph{special} bag, composed of a single set containing all such vertices; note that if a vertex of $\{x_1,\dots,x_k\}$ is a pole of $\mu$, then it belongs to the special bag. See Fig.~\ref{fig:bags} for the bags of the node in Fig.~\ref{fig:non-traversable}.

When $\mu$ is traversable, instead, the outer face of $\restrpert{\mu}$ corresponds to a single face of $\mathcal{H}$ in any planar embedding $\mathcal{G}$ of $G$. Thus, we do not need to maintain any information about the relative positions of $x_1,\dots,x_k$, and we can place all of them in the special bag. An illustration of the bags of the node represented in Fig.~\ref{fig:traversable} is given in Fig.~\ref{fig:bags-traversable}.

If the visit of the root $\rho$ of $\mathcal{T}$ at the end of the bottom-up traversal is completed without declaring $\langle G, W \rangle$ as negative, we have that $\pert{\rho} = G$ admits a planar embedding satisfying R.\ref{r:1} and thus $\langle G, W \rangle$ is a positive instance.

As anticipated in Section~\ref{sec:formulation}, we discuss two main problems arising when extending the algorithm in~\cite{DBLP:journals/jda/AngeliniBFPR12} for {\sc SEFE} to solve our problem when $H$ is not biconnected. 

First, when $H$ is biconnected it is always possible to decide the flip of every child component for every node that is either an R- or a P-node, but not when it is an S-node. On the other hand, the fact that no two S-nodes can be adjacent to each other in the SPQR-tree ensures that this choice is always fixed in the next step of the algorithm (refer to \emph{visible nodes} in~\cite{DBLP:journals/jda/AngeliniBFPR12}). When $H$ is not biconnected, even the flips of the children of R- and P-nodes (and of S-nodes) may be not uniquely determined. So, there is no guarantee that a choice for these flips can be done in the next step; in fact, it is sometimes necessary to defer this choice till the end of the algorithm. This comes with another difficulty. In the course of the algorithm, it could be required to make ``partial'' choices for these flips, in the sense that constraints imposed by the structure of the graph could enforce two or more components to be flipped in the same way (without enforcing, however, a specific flip for them). To encode the possible flips of the components that are enforced by the constraints considered till a certain point of the algorithm, we introduced the bags, which represent the main technical contribution of this work.

Second, the order of the vertices along the faces of $H$ is not unique if $H$ is not biconnected. For \embproblem, this is not an issue, as it is enough that the vertices belonging to the pairs in $W$ share a face, but we do not impose any requirement on their order along these faces. On the other hand, if we were able to also control these orders, we could provide an algorithm for instances of {\sc Sefe} in which one of the two graphs is biconnected, which would be a significant step ahead in the state of the art for this problem. We recall that an efficient algorithm for this case (even with the additional restriction that the common graph is connected) would imply an efficient algorithm for all the instances in which the common graph is connected (and no restriction on the two input graphs), as observed in~\cite{adn-osnsp-14}.

\subsection{Detailed Description of the Algorithm}\label{se:detailed}

We give the details of the algorithm. Let $\mathcal{T}$ be the SPQR-tree of $G$, rooted at a Q-node~$\rho$. First, we compute for each node $\mu \in \mathcal{T}$, whether $\mu$ is traversable or not, that is, whether there exist two paths composed of edges of $H$ between the poles of $\mu$, one in $\pert{\mu}$ and one in $G \setminus \pert{\mu}$. A na\"{i}ve approach would be to perform a BFS-visit restricted to the edges of $H$ in each of the two graphs in linear time per node, and thus in total quadratic time. For a linear-time algorithm, we proceed as follows; see also~\cite{DBLP:journals/talg/AngeliniBFJKPR15}. We traverse $\mathcal{T}$ bottom-up to compute for each node $\mu$ whether there exists the desired path in $\pert{\mu}$, using the same information computed for its children. Then, with a top-down traversal, we search for the path in $G \setminus \pert{\mu}$, using the information computed in the first traversal.

The main part of our algorithm consists of a bottom-up traversal of $\mathcal{T}$. For a node $\mu \in \mathcal{T}$, let $\langle x_1, y_1 \rangle, \dots, \langle x_k, y_k \rangle$ be all pairs of $W$ such that $x_i \in \pert{\mu}$ and $y_i \notin \pert{\mu}$. We denote by $B_\mu^1, \dots, B_\mu^q$ the bags of $\mu$ and by $\mathcal{B}_\mu$ its special bag; these bags determine a partition of the vertices $x_1,\dots,x_k$ that are required to be on the outer face of $\restrpert{\mu}$ due to R.\ref{r:2}. The vertices of each bag $B_\mu^i=\langle S_\mu^i,T_\mu^i \rangle$ are partitioned into its two pockets $S_\mu^i$ and $T_\mu^i$; all vertices of $S_\mu^i$ must lie in the same face of $\mathcal{H}$, either $f_\mu^l$ or $f_\mu^r$, while all vertices of $T_\mu^i$ must lie on the other face.

We first describe an operation, called \emph{\operation}, to modify the bags of a node $\mu$ in order to satisfy the constraints that may be imposed by R.\ref{r:1} when there exists a pair $\langle x,y \rangle \in W$ such that $x,y \in \pert{\mu}$. Refer to Figs.~\ref{fig:operation-before}--~\ref{fig:operation-after}.
In particular, if at least one of $x$ and $y$ belongs to the special bag $\mathcal{B}_\mu$ (see $\langle x_4, x_6 \rangle$ in the figure), or if $x$ and $y$ belong to the same pocket of a bag $B_\mu^i$, with $1 \leq i \leq q$, then we do not modify any bag.
If $x \in S_\mu^i$ and $y \in T_\mu^i$, for some $1 \leq i \leq q$, or vice versa, then we declare the instance negative.
Otherwise, we have $x \in B_\mu^i$ and $y \in B_\mu^j$, for some $1 \leq i \neq j \leq q$, and we merge $B_{\mu}^i$ and $B_{\mu}^{j}$ into a single bag $B_{\mu}=\langle S_\mu,T_\mu \rangle$, i.e., we merge into $S_\mu$ the pockets of $B_{\mu}^i$ and $B_{\mu}^{j}$ containing $x$ and $y$, respectively, and we merge into $T_\mu$ the other two pockets of $B_{\mu}^i$ and $B_{\mu}^{j}$; see $\langle x_2, x_5 \rangle$ in the figure.
We finally remove $\langle x,y \rangle$ from $W$ and, if there is no other pair in $W$ containing $x$ (resp., $y$), we remove it from the bag it belongs to.

\begin{figure}[t!]
	\centering
	\subfloat[\label{fig:operation-before}]{\includegraphics[width=0.24\textwidth,page=1]{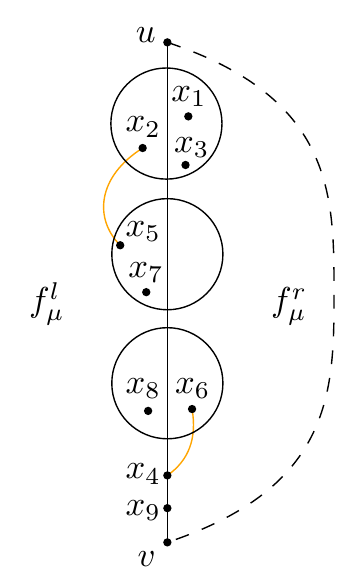}}
	\hfill
	\subfloat[\label{fig:operation-after}]{\includegraphics[width=0.24\textwidth,page=2]{images/operation.pdf}}
	\hfill
	\subfloat[\label{fig:series-before}]{\includegraphics[width=0.24\textwidth,page=1]{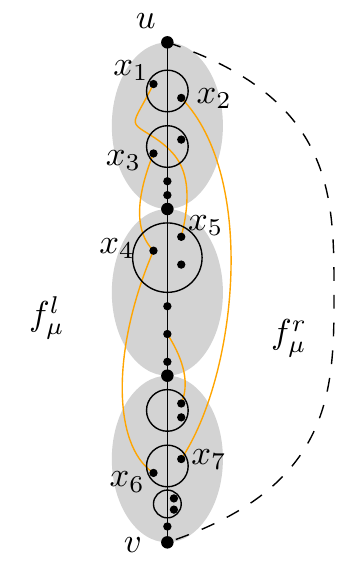}}
	\hfill
	\subfloat[\label{fig:series-after}]{\includegraphics[width=0.24\textwidth,page=2]{images/nodes.pdf}}
	\caption{
	(a) The bags of a node $\mu$ and two pairs $\langle x_2,x_5 \rangle, \langle x_4,x_6 \rangle \in W$ (orange curves). 
	(b) The bags of $\mu$ after operation \operation. Pair $\langle x_2,x_5 \rangle$ merged two bags, while $\langle x_4,x_6 \rangle$ did not modify any bag, since $x_4 \in \mathcal{B}_\mu$. 
	(c) Initialization of the bags of an S-node $\mu$. (d) The bags of $\mu$ after \operation. The instance is negative, as pair $\langle x_1,x_5 \rangle$ is such that $x_1 \in S_\mu^1$ and $x_5 \in T_\mu^1$.}
	\label{fig:operation}
\end{figure}

At each step of the traversal of $\mathcal{T}$, we consider a node $\mu$, with poles $u$ and $v$, and children $\nu_1,\dots,\nu_h$ in $\mathcal{T}$. We denote by $e_i$, for $i=1,\ldots,h$, the virtual edge of $\skel{\mu}$ corresponding to $\nu_i$.

\algocase{Suppose that $\mu$ is a Q-node.} If any of the two poles of $\mu$ belongs to $\{x_1,\dots,x_k\}$, then we add it to $\mathcal{B}_\mu$, independently of whether $\mu$ is traversable or not.

\algocase{Suppose that $\mu$ is an S-node}. We initialize special bag  $\mathcal{B}_\mu$ to the union of the special bags of $\nu_1,\dots,\nu_h$. Note that if $\mu$ is traversable, then all of its children are traversable. So, in this case, we already have that all vertices $x_1,\dots,x_k$ are in $\mathcal{B}_\mu$. Further, if $\mu$ is non-traversable, we add to the set of bags of $\mu$ all the non-special bags of its children; see Fig.~\ref{fig:series-before}.  Finally, as long as there exists a pair $\langle x,y \rangle \in W$ such that both $x$ and $y$ belong to $\pert{\mu}$, we apply operation \operation to $\langle x,y \rangle$. This may result in uncovering a negative instance, but only when $\mu$ is non-traversable. See Fig.~\ref{fig:series-after}. 

\algocase{Suppose that $\mu$ is an R-node.} See Fig.~\ref{fig:rigid}.  
Let $\subskel$ be the graph composed of the vertices of $\skel{\mu}$ and of the virtual edges corresponding to non-traversable children of $\mu$, plus $\redge{\mu}$ if $\mu$ is non-traversable; see Fig.~\ref{fig:subskel}. Let $\subskelemb$ be the restriction of the unique planar embedding of the triconnected graph $\skel{\mu}$ to $\subskel$. 
Note that, for each traversable child $\nu_i$ of $\mu$, virtual edge $e_i$ is \emph{contained in} one face $f_{\nu_i}$ of $\subskelemb$; in Fig.~\ref{fig:subskel}, $(w_4,w_6)$ is contained in face $\{w_3,w_4,w_5,w_6\}$. For a non-traversable child $\nu_i$, denote by $f_{\nu_i}^1$ and $f_{\nu_i}^2$ the two faces of $\subskelemb$ virtual edge $e_i$ is incident to. For a vertex $x \in V$ that does not belong to $\skel{\mu}$, we denote by $\e{x}$ either the virtual edge $e_i$, if $x \in \pert{\nu_i}$, or the virtual edge $\redge{\mu}$ representing the parent of $\mu$, if $x \in G \setminus \pert{\mu}$.

Suppose that $\mu$ is non-traversable; see Fig.~\ref{fig:association}. Recall that in this case $\redge{\mu} \in \subskel$; let $f_\mu^l$ and $f_\mu^r$ be the two faces of $\subskelemb$ incident to $\redge{\mu}$. Any other virtual edge $e_i$ of $\subskel$ such that $\{f_{\nu_i}^1,f_{\nu_i}^2\} = \{f_\mu^l,f_\mu^r\}$ is called \emph{$2$-sided}; see $(w_2,w_3),(v,w_6)$ in Fig.~\ref{fig:association}. 

We consider each pair $\langle x,y \rangle \in W$ such that $x \in \pert{\nu_i}$, with $1 \leq i \leq h$, and $y \notin \pert{\nu_i}$. Let $e^x=\e{x}$ and $e^y=\e{y}$. A necessary condition for R.\ref{r:1} and R.\ref{r:2} is that $e^y$ is either contained in or incident to face $f_{\nu_i}$ (if $\nu_i$ is traversable) or one of $f_{\nu_i}^1$ and $f_{\nu_i}^2$ (if $\nu_i$ is non-traversable). If this is not the case, we declare the instance negative. 

\begin{figure}[t!]
	\centering
	\subfloat[\label{fig:rigid}]{\includegraphics[width=0.24\textwidth,page=3]{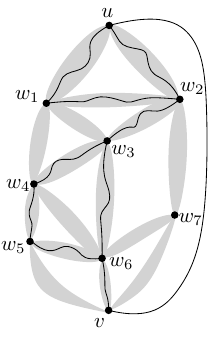}}
	\hfill
	\subfloat[\label{fig:subskel}]{\includegraphics[width=0.24\textwidth,page=4]{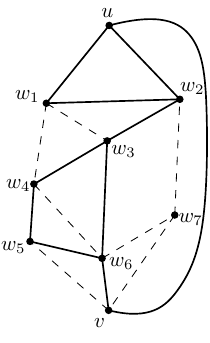}}
	\hfill
	\subfloat[\label{fig:association}]{\includegraphics[width=0.24\textwidth,page=5]{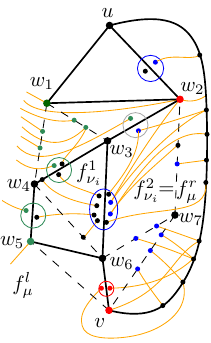}}
	\hfill
	\subfloat[\label{fig:bags-rigid}]{\includegraphics[width=0.24\textwidth,page=6]{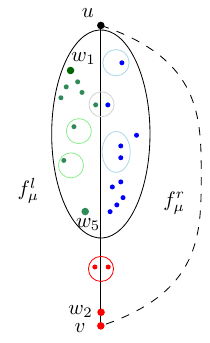}}
	\caption{
	(a) Graph $\pert{\mu}$ when $\mu$ is a non-traversable R-node. 
	(b) Graph $\subskel$ (solid) and the traversable children (dashed) of $\mu$. 
	    Children corresponding to virtual edges $(v,w_6)$ and $(w_2,w_3)$ are $2$-sided.
	(c) Association of pockets with faces. 
	    Blue (green) pockets are associated with $f_\mu^r$ ($f_\mu^l$, resp.).
	    Red pockets are not associated.
	    Gray pockets belong to $2$-sided children, but they are associated with $f_\mu^r$ and $f_\mu^l$.
	(d) The bags of $\mu$.}
	\label{fig:rigid-node}
\end{figure}

Another constraint imposed by this pair is the following.
Suppose that $x$ belongs to a pocket, say $S_{\nu_i}$, of a bag $B_{\nu_i}$ of $\nu_i$ (this can only happen if $\nu_i$ is non-traversable). 
If $e^x$ and $e^y$ share exactly one face, say $f_{\nu_i}^1$, then all pairs $\langle x', y' \rangle \in W$ with $x' \in S_{\nu_i}$ must be such that $\e{y'}$ is either contained in or incident to $f_{\nu_1}$; also, all the pairs $\langle x'', y'' \rangle \in W$ with $x'' \in T_{\nu_i}$ must be such that $\e{y''}$ is either contained in or incident to $f_{\nu_2}$. This is due to the fact all the vertices in the same pocket must be incident to the same face of $\subskelemb$. So, if this is not the case, we declare the instance negative. Otherwise, we \emph{associate} $S_{\nu_i}$ with $f_{\nu_i}^1$ and $T_{\nu_i}$ with $f_{\nu_i}^2$. 
If $e^x$ and $e^y$ share both faces $f_{\nu_i}^1$ and $f_{\nu_i}^2$, instead, we have to postpone the association of $S_{\nu_i}$ and $T_{\nu_i}$, as at this point we cannot make a unique choice. Note that an association for these pockets may be performed later, due to another pair of $W$.
Suppose now that $x$ belongs to the special bag $\mathcal{B}_{\nu_i}$ of ${\nu_i}$. Then, we associate $\mathcal{B}_{\nu_i}$ to either $f_{\nu_i}$, if $\nu_i$ is traversable, or to both $f_{\nu_i}^1$ and $f_{\nu_i}^2$, if it is non-traversable. 
This completes the process of pair $\langle x,y \rangle$.

Once all children $\nu_1,\dots,\nu_h$ of $\mu$ have been considered, there may still exist pockets that are not associated. Let $S_{\nu_i}$ be one of such pockets, and consider each pair $\langle x,y \rangle \in W$ such that $x \in S_{\nu_i}$. Note that $\e{x}$ shares both faces $f_{\nu_i}^1$ and $f_{\nu_i}^2$ with $\e{y}$. If $y$ belongs to a pocket, say $T_{\nu_j}$, that is associated with one of $f_{\nu_i}^1$ and $f_{\nu_i}^2$, say $f_{\nu_i}^1$, then we associate $S_{\nu_i}$ with $f_{\nu_i}^1$ and $T_{\nu_i}$ with $f_{\nu_i}^2$. In fact, the association of $T_{\nu_j}$ with $f_{\nu_i}^1$ implies that $y$ will be incident to $f_{\nu_i}^1$ in any embedding of $G$ that is a solution for $\langle G,W \rangle$. If two pairs determine different associations for $S_{\nu_i}$ and $T_{\nu_i}$, we declare the instance negative. 

We repeat the above process as long as there exist pockets that can be associated by means of this procedure. Note that this does not necessarily result in an association for all pockets; however, we can say that all the mandatory choices for $\pert{\mu}$ have been performed. Consider any of the remaining pockets $S_{\nu_i}$. If $\nu_i$ is not $2$-sided, then we associate $S_{\nu_i}$ with $f_{\nu_i}^1$ and $T_{\nu_i}$ with $f_{\nu_i}^2$. This association can be done arbitrarily since its effect is limited to $\pert{\mu}$ and not to $G \setminus \pert{\mu}$, as $\nu_i$ is not $2$-sided. Then, we propagate this association to other pockets by performing the procedure described above. We repeat this process until the only pockets that are not associated, if any, belong to bags of $2$-sided children of $\mu$. Note that the previous arbitrary association cannot be propagated to pockets of $2$-sided children, since their virtual edges are only incident to $f_{\mu}^l$ and $f_{\mu}^r$.

Based on the association of $\nu_1,\dots,\nu_h$ with the faces of $\subskelemb$, we determine the bags of $\mu$; see Fig.~\ref{fig:bags-rigid}. The special bag $\mathcal{B}_\mu$ of $\mu$ contains the poles of $\mu$, if they belong to $\{x_1,\dots,x_k\}$, and the union of the special bags of the $2$-sided children of $\mu$. Next, we create a bag $B_{\mu}=\langle S_{\mu},T_{\mu} \rangle$, such that $S_{\mu}$ and $T_{\mu}$ contain all the vertices of the pockets associated with $f_\mu^l$ and $f_\mu^r$, respectively. 
Finally, we add to the set of bags of $\mu$ the non-special bags of the $2$-sided children of $\mu$ whose pockets have not been associated with any face of $\subskelemb$ (this allows us to postpone their association). Then, we apply operation \operation to all pairs $\langle x,y \rangle \in W$ such that both $x$ and $y$ belong to $\pert{\mu}$ in order to merge the bags of different $2$-sided children of $\mu$ (again this may result in uncovering a negative instance). This completes the case in which $\mu$ is non-traversable.

It remains to consider the simpler case in which $\mu$ is traversable. In this case virtual edge $\redge{\mu}$ does not belong to $\subskel$; hence faces $f_{\mu}^l$ and $f_{\mu}^r$ do not exist, and none of the children of $\mu$ is $2$-sided. This implies that performing all the operations described above results in an association of each pocket and of each special bag of the children of $\mu$ with some face of $\subskelemb$. Recall that, since $\mu$ is traversable, $\mu$ has only its special bag $\mathcal{B}_\mu$. We add to $\mathcal{B}_\mu$ all the vertices of the pockets and of the special bags that have been associated with the outer face of $\subskelemb$. This concludes the R-node case.

\algocase{Suppose that $\mu$ is a P-node.} Refer to Fig.~\ref{fig:parallel-node}. We distinguish three cases, based on whether $\mu$ has%
\begin{inparaenum}[(i)]
\item \label{p:0} zero, 
\item \label{p:1} one, or 
\item \label{p:2} more than one non-traversable child.
\end{inparaenum}

In Case~(\ref{p:0}), we have that $\mu$ is traversable. So, it has only its special bag $\mathcal{B}_\mu$, in which we add all the vertices of the special bags of its children. Note that all virtual edges in $\skel{\mu}$ are incident to the same face of $\restrpert{\mu}$ and hence R.\ref{r:1} and R.\ref{r:2} are trivially satisfied.

\begin{figure}[t]
	\centering
	\subfloat[\label{fig:parallel}]{\includegraphics[width=0.18\textwidth,page=7]{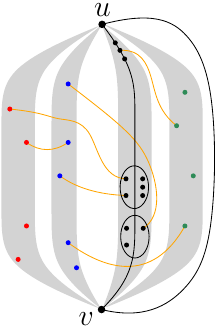}}
	\hfil
	\subfloat[\label{fig:parallel-1}]{\includegraphics[width=0.18\textwidth,page=8]{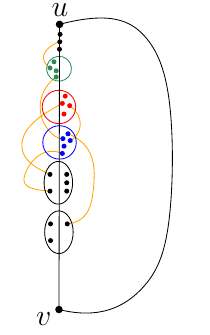}}
	\hfil
	\subfloat[\label{fig:parallel-2}]{\includegraphics[width=0.18\textwidth,page=9]{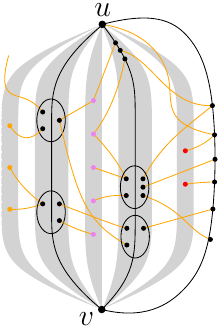}}
	\hfil
	\subfloat[]{\includegraphics[width=0.18\textwidth,page=10]{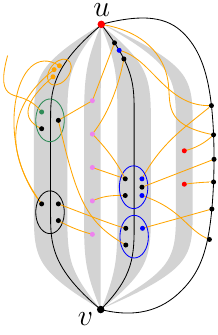}}
	\hfil
	\subfloat[\label{fig:bags-parallel}]{\includegraphics[width=0.18\textwidth,page=11]{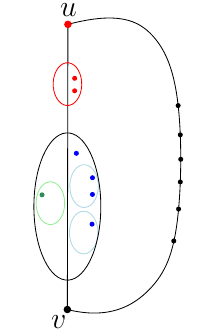}}
	\caption{Illustration for the case in which $\mu$ is a P-node with: 
	(a-b)~one, and (c-e)~more than one non-traversable children.
	The color-scheme of this figure follows the one of Fig.~\ref{fig:rigid-node}}
	\label{fig:parallel-node}
\end{figure}

Next, we consider Case~(\ref{p:1}), in which $\mu$ has exactly one non-traversable child, say~$\nu_1$; see Fig.~\ref{fig:parallel}. In this case, $\mu$ is non-traversable, since the path of $G \setminus \pert{\nu_1}$ composed of edges of $H$ also belongs to $G \setminus \pert{\mu}$. We initialize the set of bags of $\mu$ to the set of bags of $\nu_1$. For each traversable child $\nu_i$, with $i=2,\dots,h$, we add to $\mu$ a new bag $B_{\mu}^i$, where $S_{\mu}^i$ contains all the vertices in the special bag $\mathcal{B}_{\nu_i}$ of $\nu_i$, while $T_{\mu}^i$ is empty; see Fig.~\ref{fig:parallel-1}. This represents the fact that all the vertices in $\pert{\nu_i}$ must lie on the same side of the cycle passing through $\pert{\nu_1}$ and $G \setminus \pert{\mu}$ to satisfy R.\ref{r:2}. Finally, we apply operation \operation to all pairs $\langle x,y \rangle \in W$ such that both $x$ and $y$ belong to $\pert{\mu}$. 

Finally, we consider Case~(\ref{p:2}), in which $\mu$ has more than one non-traversable child; see Figs.~\ref{fig:parallel-2}-\ref{fig:bags-parallel}. We construct an auxiliary graph $\aux$ with a vertex $v_i$ for each child $\nu_i$ of $\mu$, which is colored \emph{black} if $\nu_i$ is non-traversable and \emph{white} otherwise. Graph $\aux$ also has a vertex $v$ corresponding to $\redge{\mu}$, which is colored black if $\mu$ is non-traversable and white otherwise. Then, we consider every pair $\langle x,y \rangle \in W$ such that $x \in \pert{\nu_i}$, for some child $\nu_i$ of $\mu$. If $y \in \pert{\nu_j}$, for some $j \neq i$, then we add edge $(v_i,v_j)$ to $\aux$, while if $y \in G \setminus \pert{\mu}$, then we add edge $(v_i,v)$ to $\aux$. If $\aux$ has multiple copies of an edge, we keep only one of them. We assume w.l.o.g.\ that no two white vertices are adjacent in $\aux$, as otherwise we could contract them to a new white vertex. In fact, the virtual edges representing traversable children of $\mu$ corresponding to adjacent white vertices must be contained in the same face of $\restrpert{\mu}$, due to R.\ref{r:1}.

Consider each white vertex $w$ of $\aux$. If $w$ has more than two black neighbors, we declare the instance negative, as the virtual edge of the traversable child of $\mu$ corresponding to $w$ should share a face in $\restrpert{\mu}$ with more than two virtual edges representing non-traversable children of $\mu$, which is not possible. If $w$ has at most one black neighbor, we remove $w$ from $\aux$. Finally, if $w$ has exactly two black neighbors $b$ and $b'$, then we remove $w$ from $\aux$ and we add edge $(b,b')$ to $\aux$ (if it is not present). Once we have considered all white vertices, the resulting graph  $\auxb$ has only black vertices. 

We check whether $\auxb$ is either a cycle through all its vertices or a set of paths (some of which may consist of single vertices). The necessity of this condition can be proved similar to~\cite{DBLP:journals/jda/AngeliniBFPR12}. The only difference is in the edges between black vertices that are introduced due to degree-$2$ white vertices. Let $(b,b')$ be one of such edges and let $w$ be the white vertex that was adjacent to $b$ and $b'$. Also, let $e_b$, $e_{b'}$, and $e_w$ be the virtual edges representing the children of $\mu$ (or virtual edge $\redge{\mu}$, if $\mu$ is non-traversable) corresponding to $b$, $b'$, and $w$, respectively. Then, $e_b$ and $e_{b'}$ must share a face in $\restrpert{\mu}$, and this face must contain $e_w$, due to R.\ref{r:1} and R.\ref{r:2}. If the above condition on $\auxb$ is not satisfied, then we declare the instance negative; otherwise, we fix an order of the black vertices of $\auxb$ based either on the cycle or on an arbitrary order of the paths.

We now construct graph $\subskel$ in the same way as for the R-node. Note that, also in this case, the embedding $\subskelemb$ of $\subskel$ is fixed, since the order of the black vertices of $\auxb$ induces an order of the virtual edges of $\subskel$. We will again use $\subskelemb$ to either determine whether the instance is negative or to construct the bags of $\mu$.

The case in which $\mu$ is traversable is identical to the R-node case. When $\mu$ is non-traversable, we have $\redge{\mu} \in \subskel$, and thus there exist the two faces $f_{\mu}^l$ and $f_{\mu}^r$ incident to $\redge{\mu}$. However, since $\mu$ has at least two non-traversable children, every two virtual edges of $\subskel$ share at most one face in $\subskelemb$, and thus $\mu$ has no $2$-sided children.

We now consider each traversable child $\nu_i$ of $\mu$. Contrary to the R-node case, the face of $\subskelemb$ in which $\nu_i$ is contained is not necessarily defined in this case by the rigid structure underneath, as the embedding of $\skel{\mu}$ is not unique. Recall that $\nu_i$ corresponds to a white vertex $v_i$ of $\aux$. If $v_i$ has exactly two black neighbors, then they must be connected by an edge in $\auxb$ after the removal of $v_i$. So, they are consecutive in the order of the black vertices that we used to construct $\subskelemb$. Thus, the two virtual edges of $\subskel$ corresponding to them share a face in $\subskelemb$. We say in this case that $e_i$ is \emph{contained in} this particular face. If $v_i$ has exactly one black neighbor in $\aux$, then $e_i$ may be contained in any of the two faces of $\subskelemb$ incident to the virtual edge $e$ corresponding to this black vertex. However, we cannot make a choice at this stage, as this may depend on other pairs whose vertices belong to the subgraph of $G$ represented by $e$ (that is, $\pert{\nu_j}$, if $e = e_j$, for some $1 \leq j \leq h$, and $G \setminus \pert{\mu}$, if $e = \redge{\mu}$). If $e = e_j$, then we add a new bag $B_{\nu_j}$ to the child $\nu_j$ of $\mu$, so that $S_{\nu_j}$ contains all the vertices of the special bag of $\nu_i$, while $T_{\nu_j}$ is empty. The association of $S_{\nu_j}$ with one of the two faces incident to $e_j$, to be performed later, will determine the face in which $e_i$ is contained. In the case in which $e = \redge{\mu}$, virtual edge $e_i$ should be contained either in $f_{\mu}^l$ or in $f_{\mu}^r$, but again we cannot determine which of the two. Furthermore, we cannot even delegate this choice to the association of the pockets, since $\redge{\mu}$ does not correspond to a child of $\mu$. Thus, we do not associate it to any face, but we will use it to create the bags of $\mu$. Finally, when $v_i$ has no black neighbors, its special bag is empty. 

Once all traversable children have been considered, we associate the special bags and the pockets of the non-special bags with the faces of $\subskelemb$, as in the R-node case. Then, we construct the bags of $\mu$. We add the poles of $\mu$ to its special bag, if they belong to $\{x_1,\dots,x_k\}$. As in the R-node case, we add to $\mu$ a bag $B_{\mu}$, whose pockets $S_{\mu}$ and $T_{\mu}$ have all the vertices of the special bags and of the pockets associated with $f_{\mu}^l$ and $f_{\mu}^r$, respectively. Finally, for each traversable child $\nu_i$ of $\mu$ that has not been associated, we add a new bag $B_{\mu}^i$ so that $S_{\mu}^i$ contains all the vertices of the special bag of~$\nu_i$, while $T_{\mu}^i$ is empty.
Finally, we apply operation \operation to all pairs $\langle x,y \rangle \in W$ such that both $x$ and $y$ belong to $\pert{\mu}$. Hence, R.\ref{r:1} and R.\ref{r:2} are satisfied by any embedding $\embpert{\mu}$ of $\pert{\mu}$ that is described by the bags of $\mu$. This concludes the P-node case.

At the end of the traversal, if root $\rho$ has been visited without declaring the instance negative, the fact that $\pert{\rho} = G$ admits a planar embedding satisfying R.\ref{r:1} implies that $\langle G, W \rangle$ is a positive instance. The proof of the following theorem is in the appendix.

\begin{theorem}\label{th:algorithm-biconnected}
Let \ourinstance{} be an instance of \ourproblem such that the graph induced by the edges in $\Ef \cup \Ei$ is biconnected. We can test in $O(|V|^3 \cdot |\En|)$ time whether $G$ has a drawing with no forbidden crossing.
\end{theorem}
\begin{proof}
By Theorem~\ref{th:equivalence}, it suffices to prove that the algorithm described in Section~\ref{se:detailed} decides in $O(|V|^3 \cdot |W|)$ whether an instance \embinstance{} of \embproblem in which the graph $H$ induced by the edges of $E_1$ is biconnected is positive.
	
The correctness of the algorithm follows from the fact that, as already discussed during the description of the algorithm, for each node $\mu \in \mathcal{T}$, requirements R.\ref{r:1} and R.\ref{r:2} are satisfied by any embedding $\embpert{\mu}$ of $\pert{\mu}$ that is described by the bags of $\mu$ (if any). In particular, this holds also for the root $\rho$ of $\mathcal{T}$.
	
Regarding the time complexity, we observe that the construction of the SPQR-tree $\mathcal{T}$ and of the auxiliary graphs $\aux$ and $\subskel$ can be done in $O(|V|+|W|)$ time. Operation \operation needs constant time, adopting elementary data structures to maintain the references between vertices and bags or pockets. Thus, the complexity of our algorithm is dominated by the association of the bags and of the pockets to the faces of $\subskelemb$, in the R- and P-node cases. In this phase of the algorithm, every bag of a child of a node $\mu$ could be considered a number of times that is linear in the total number of bags, which is $O(|V|)$. Also, every time one of these bags is considered, we perform $O(|W|)$ checks. Since the number of bags over all the children of $\mu$ is $O(|V|)$, we have a total $O(|V|^2 \cdot |W|)$ processing time for $\mu$, which hence results in a total $O(|V|^3 \cdot |W|)$ time for $G$, and the statement follows.
\end{proof}

\section{Conclusions}
\label{sec:conclusions}

In this paper we studied the problem \ourproblem, in which a graph whose edges are of three types (\fundamental, \important, and \nonimportant) is given and the goal is to construct a drawing in which crossings are allowed only if they involve a \nonimportant edge. For this problem, we gave an efficient algorithm when the graph induced by the \fundamental and \important edges is biconnected. 

The main open problem raised by our work is to determine the complexity in the general case, where the biconnectivity restriction is relaxed. It is also of interest to broaden the study towards the case in which there exist more than three levels of importance for the edges. As a first step, one could consider the case in which there are four levels and the first two form a biconnected graph. Finally, the relationship with {\sc Sefe} should be further investigated to understand whether the techniques used in this paper can be applied to solve some of its open cases.

\bibliographystyle{abbrvurl}
\bibliography{references}

\end{document}